\documentclass[a4paper,USenglish]{lipics-v2018}

\usepackage{microtype}%if unwanted, comment out or use option "draft"

\title{Characterizing Asynchronous Message-Passing Models Through Rounds}

\author{Adam Shimi}{IRIT -- Université de Toulouse,
                    2 rue Camichel,
                    F-31000 Toulouse, France\\
                    \texttt{http://www.irit.fr}}{adam.shimi@irit.fr}{}{}
\author{Aurélie Hurault}{IRIT -- Université de Toulouse,
                         2 rue Camichel,
                         F-31000 Toulouse, France\\
                         \texttt{http://www.irit.fr}}{aurelie.hurault@irit.fr}{}{}
\author{Philippe Quéinnec}{IRIT -- Université de Toulouse,
                           2 rue Camichel,
                           F-31000 Toulouse, France\\
                           \texttt{http://www.irit.fr}}{philippe.queinnec@irit.fr}{}{}

\authorrunning{A. Shimi, A. Hurault and P. Quéinnec} %mandatory. First: Use abbreviated first/middle names. Second (only in severe cases): Use first author plus 'et. al.'

\Copyright{Adam Shimi, Aurélie Hurault and Philippe Quéinnec}%mandatory, please use full first names. LIPIcs license is "CC-BY";  http://creativecommons.org/licenses/by/3.0/

\subjclass{}% mandatory: Please choose ACM 1998 classifications from http://www.acm.org/about/class/ccs98-html . E.g., cite as "F.1.1 Models of Computation".
% mandatory: Please provide 1-5 keywords

\funding{This work was supported by project PARDI ANR-16-CE25-0006.}

\EventEditors{Jiannong Cao, Faith Ellen, Luis Rodrigues, and Bernardo Ferreira}
\EventNoEds{4}
\EventLongTitle{22st International Conference on Principles of Distributed Systems (OPODIS 2018)}
\EventShortTitle{OPODIS 2018}
\EventAcronym{OPODIS}
\EventYear{2018}
\EventDate{December 17--19, 2018}
\EventLocation{Hong Kong, China}
\EventLogo{}
\SeriesVolume{122}
\ArticleNo{0} % The number will be adjusted after you submit your camera ready latex sources
\nolinenumbers

  \newtheorem{lemmabis}{Lemma}

\begin{document}

\maketitle

\begin{abstract}
    Message-passing models of distributed computing vary along numerous
    dimensions: degree of synchrony, kind of faults, number of faults...
    Unfortunately, the sheer number of models and their subtle distinctions
    hinder our ability to design a general theory of message-passing models.
    One way out of this conundrum restricts communication to proceed by
    round. A great variety of message-passing models can then be captured in
    the Heard-Of model, through predicates on the messages sent in a round
    and received during or before this round. Then, the issue is to find the
    most accurate Heard-Of predicate to capture a given model. This is
    straightforward in synchronous models, because waiting for the upper
    bound on communication delay ensures that all available messages are
    received, while not waiting forever. On the other hand, asynchrony
    allows unbounded message delays. Is there nonetheless a meaningful
    characterization of asynchronous models by a Heard-Of predicate?

    We formalize this characterization by introducing
    Delivered collections: the collections of all messages delivered
    at each round, whether late or not. Predicates on
    Delivered collections capture message-passing models.
    The question is to determine which Heard-Of predicates can be generated
    by a given Delivered predicate. We answer this by formalizing
    strategies for when to change round. Thanks to a partial order on these
    strategies, we also find the "best" strategy for multiple models,
    where "best" intuitively means it waits for as many messages as possible
    while not waiting forever. Finally, a strategy for changing round that
    never blocks a process forever implements a Heard-Of predicate. This
    allows us to translate the order on strategies into an order on Heard-Of
    predicates.
    The characterizing predicate for a model is then the greatest element
    for that order, if it exists.
\end{abstract}

\section{Introduction}
\label{sec:intro}

    \subsection{Motivation}
    \label{subsec:motivation}

        Even when restricted to message-passing,
        distributed computing spawns a plethora of models:
        with various degrees of synchrony,
        with different kinds of faults, with different failure detectors...
        Although some parameters are quantitative, such as the number of faults,
        the majority
        are qualitative instead, for example the kinds of faults.
        Moreover, message-passing models
        are usually defined by a mix of mathematical formalism and textual
        description,
        with crucial details nested deep inside the latter.
        This is why these models resist unification into a theory
        of distributed computing, and why results in the field are
        notoriously hard to organize, use and extend.

        One solution requires constraining communication
        to proceed by round: each process repeatedly broadcasts a message
        with its current round number, waits for as many messages as possible
        bearing this round number, and changes round by computing both
        its next state and next message.
        The variations between models are then captured by the dynamic
        graph specifying, for each round, from which processes each process received
        a message with this round number before the end of its round;
        this fits the concept of dynamic
        network from Kuhn and Oshman~\cite{KuhnDynamic}. Nonetheless,
        we will privilege the perspective of Charron-Bost and Schiper
        Heard-Of model~\cite{CharronBostHO}, which places itself more at the
        level of processes. Here, the Heard-Of collection of an execution
        contains, for each round $r$ and each process $j$,
        the set of processes from which $j$ received a message sent in round
        $r$ before going to round $r+1$. Then, a predicate on Heard-Of collections
        characterizes a message-passing model.

        Yet rounds don't remove the complexities and subtleties of
        message-passing models -- they just shift them to the characterization
        of a given model by a Heard-Of predicate. This characterization depends
        on how rounds can be implemented in the underlying model.
        In the synchronous case, processes progress in lock-step, and every message
        that will ever be received is received during its corresponding round.
        Hence, the Heard-Of predicate characterizing a synchronous model
        simply specifies which messages can be lost.
        In asynchronous models on the other hand, messages can be late, and thus
        the distance between the round numbers of processes is unbounded.
        The combination of these uncertainties implies that processes do not
        know which messages will be delivered and when. Thus, there is a risk
        of not waiting for useful messages that will eventually arrive,
        and to wait forever for messages that never will.

        To the best of our knowledge, there is no systematic study of
        the Heard-Of predicates generated by various asynchronous
        message-passing models. Because it is a crucial step in unifying
        distributed computing's menagerie of models through rounds,
        we also believe this topic to be of importance.

    \subsection{Approach and Overview}
    \label{subsec:approachOverview}

        As hinted above, the difficulty lies in the potential discrepancy
        between messages delivered on time -- captured by Heard-Of collections --
        and messages delivered at all. We want to determine the former, but
        it is considerably easier to specify the latter for an operational
        model: list the messages that will eventually arrive.
        We therefore center our formalization around Delivered
        collections, infinite sequences of communication graphs capturing,
        for each round, all messages sent at this round and eventually delivered
        for a given operational model.
        The question is thus to characterize by a Heard-Of predicate
        which messages can be waited for when the deliveries are those from
        the Delivered predicate.

        From these Delivered collections, we build runs representing
        the different scheduling of deliveries and changes of round. Some
        of these runs, called valid, define a Heard-Of collection;
        invalid runs have processes blocked forever at some round. We
        filter the latter thanks to strategies: sets of local states
        for which a process is allowed to change round. Runs for a strategy
        must also satisfy a fairness condition ensuring that if a process
        can change round continuously, it does. Strategies with
        only valid runs for a Delivered predicate,
        that is strategy implementing a Heard-Of predicate, are called
        valid.

        The next question is how to choose a valid strategy and the
        corresponding predicate, as characterizing a Delivered predicate and
        its underlying model? We answer by taking the strategy generating
        the Heard-Of predicate that is the smallest overapproximation of the
        Delivered predicate.
        From this
        intuition, we define a partial order on valid strategies called
        domination; a characterizing strategy is a greatest element for this
        order, and the characterizing predicate is the one it generates.

        The results obtained with this approach are threefold:
        \begin{itemize}
            \item The formalization itself, with a complete example: the
                asynchronous message-passing model with reliable communication
                and at most $F$ permanent crashes.
            \item The study of carefree strategies, the ones depending only
                on messages from the current round. This restricted class
                is both well-behaved enough to always have a unique dominating
                strategy, and expressive enough to capture interesting
                Delivered predicates.
            \item The study of reactionary strategies, the ones depending only
                on messages from past and current rounds.
                Here too we show well-behavior of this class as well as
                an example where reactionary is needed for domination,
                and another one where it is insufficient.
        \end{itemize}
        Along these results, we also formally prove the characterization of
        asynchronous models by Heard-Of predicates given by Charron-Bost and
        Schiper~\cite{CharronBostHO}.

        We begin by the formalization in \textbf{Section~\ref{sec:formal}},
        while \textbf{Section~\ref{sec:formalexample}} introduces
        a fully developed example: the asynchronous
        message-passing model with reliable communication and
        at most $F$ permanent crashes.
        \textbf{Section~\ref{sec:memless}}
        explores carefree strategies in terms of well-behavior
        and expressivity, closing with the example with at most $B$ failed
        broadcasts per round, a Delivered predicate dominated by
        a carefree strategy.
        We follow by studying reactionary strategies in
        \textbf{Section~\ref{sec:general}}.
        Here again, well-behavior and expressivity are examined, followed
        by the example of at most $F$ permanent initial
        crashes. \textbf{Section~\ref{sec:relatedworks}}
        and \textbf{Section~\ref{sec:conclusion}} then conclude the paper
        with a discussion of related works, the value of our results
        and some perspectives.

\section{Formalization}
\label{sec:formal}

    All our abstractions revolve around infinite sequences of
    graphs, called collections. A Delivered collection maps each round $r$
    and process $j$ to the set of processes from which $j$
    receive a message sent at $r$. A Heard-Of collection maps each round $r$ and process $j$
    to the set of processes from which $j$ received, before going to round
    $r+1$, the message sent at $r$.
    The difference lies in considering all deliveries for Delivered collections,
    but only the deliveries before the end of the round of the receiver
    for Heard-Of collections.

    \begin{definition}[Collections and Predicates]
        Let $\Pi$ a set of processes.
        $Col: (\mathbb{N}^* \times \Pi) \mapsto \mathcal{P}(\Pi)$ is
        either a \textbf{Delivered collection} or a \textbf{Heard-Of collection}
        for $\Pi$, depending on the context.

        In the same way,
        $Pred: \mathcal{P}((\mathbb{N}^* \times \Pi) \mapsto \mathcal{P}(\Pi))$
        is either a \textbf{Delivered predicate} or a \textbf{Heard-Of predicate}
        for $\Pi$.

        For a given \textit{Col}, the kernel of round $r$
        are the processes from which everyone receives a message for this round
        $K_{\textit{Col}}(r) \triangleq \bigcap\limits_{j \in \Pi} \textit{Col}(r,j)$.
    \end{definition}

    \subsection{Runs and Strategies}

    The behavior of processes is classically specified by runs,
    sequences of both states and transitions satisfying some restricting
    conditions: messages cannot be delivered before the round they are sent
    and are delivered only once. Given a Delivered collection, we
    additionally require that the delivered messages are exactly
    the ones in the Delivered collection.

    As we only care about which messages can
    be waited for, ignoring the content of messages or the underlying
    computation, we limit the state to the received messages and the round.

    \begin{definition}[Run]
        Let $\Pi$ be a set of $n$ processes.
        Let $Q = (\mathbb{N} \times \mathcal{P}(\mathbb{N}^* \times \Pi))$
        the set of process states.
        The first element is the round of a process (written $q.round$
        for $q \in Q$) and the second
        is the set of pairs $\langle$round it was sent, sender$\rangle$
        for each delivered message (written \textit{q.received} for $q \in Q$).
        Let the set of transitions $T = \{ \textit{next}_j \mid j \in \Pi \} \cup
        \{ \textit{deliver}(r,k,j) \mid r \in \mathbb{N}^* \land k,j \in \Pi\}
        \cup \{ end \}$. \textit{next}$_j$ is the transition for $j$ changing round,
        \textit{deliver}$(r,k,j)$ is the transition for the delivery
        to $j$ of the message sent by $k$ in round $r$, and \textit{end} is
        the transition to end a finite run.
        For $qt \in Q^n \times T$ and $j \in \Pi$,
        we write $qt.state$ for the state, $qt.state.j$ for the local state of $j$
        and $qt.transition$ for the transition.
        Finally, let $(Q^n \times T)^{\infty}$
        be the set of finite and infinite words on the set $(Q^n \times T)$.
        Then, $t \in (Q^n \times T)^{\infty}$ is a \textbf{run}
        $\triangleq$
        \begin{itemize}
            \item \textbf{(Initial state)}
                $t[0].state = \langle 1, \emptyset \rangle^n$
            \item \textbf{(Transitions)}
              $\forall i \in [0,\textit{size}(t)):\\
              \hspace*{-1.5em}
              \left(
              \begin{array}{l}
                  \quad \exists r \in \mathbb{N}^*,\exists k,j \in \Pi: t[i].transition = deliver(r,k,j) \implies\\
                  \quad\quad t[i+1].state = t[i].state \\
                    \quad\quad\quad\quad\quad\quad \textit{ Except }
                        t[i+1].state.j.received =
                        t[i].state.j.received \cup \{(r,k)\}\\
                  \land~ \exists j \in \Pi: t[i].transition = next_j \implies \\
                    \quad\quad
                        t[i+1].state = t[i].state \textit{ Except }
                        t[i+1].state.j.round =
                        t[i].state.j.round+1\\
                   \land~ (t[i].transition = end) \implies (i = size(t) - 1)\\
              \end{array}
              \right)$
            \item \textbf{(Delivery after sending)}
                $\forall i \in [0,\textit{size}(t)):\\
                t[i].transition = deliver(r,k,j) \implies
                    t[i].state.k.round \geq r$
            \item \textbf{(Unique delivery)}
                $\forall \langle r, k, j \rangle \in
                (\mathbb{N}^* \times \Pi \times \Pi):
                (\exists i \in [0,size(t)): t[i].transition = deliver(r,k,j))
                \implies
                (\forall i' \in [0,size(t)) \setminus \{i\}:
                t[i'].transition \neq deliver(r,k,j))$
        \end{itemize}

        Let \textit{CDel} be a Delivered collection.
        Then, $runs(\textit{CDel})$, the \textbf{runs of} \textit{CDel} $\triangleq$\\
        $\left\{
            t \textit{ a run} ~\middle|~
            \begin{array}{l}
                \forall \langle r, k, j \rangle
                \in \mathbb{N}^* \times \Pi \times \Pi:\\
                \quad(
                    k \in \textit{CDel}(r,j)
                    \land \exists i \in [0,\textit{size}(t)):
                        t[i].state.j.round \geq r
                )\\
                \quad\iff\\
                \quad(
                    \exists i \in [0,\textit{size}(t)): t[i].transition = deliver(r,k,j)
                )
            \end{array}
            \right\}$

        For \textit{PDel} a Delivered predicate, we write
        $runs(PDel) = \{runs(\textit{CDel}) \mid \textit{CDel} \in PDel\}$.
    \end{definition}

    Our definition of runs does not force
    processes to change rounds. This contradicts our
    intuition about a system using rounds: processes should keep on "forever",
    or at least as long as necessary. In a valid run, processes change round an
    infinite number of times.

    \begin{definition}[Validity]
        A run $t$ is \textbf{valid}
        \mbox{$\triangleq
                \forall j \in \Pi:
                |\{i \in \mathbb{N} \mid t[i].transition = next_j\}| = \aleph_0$}.
    \end{definition}

    Valid run are necessarily infinite. Yet the definition above allows
    finite runs thanks to the \textit{end} transition. This is used in
    proofs by contradiction which imply the manipulation of invalid runs and
    thus potentially finite ones.

    Next, we define the other building block of our approach: strategies.
    They are simply sets of local states, representing the states where
    processes can change round.

    \begin{definition}[Strategy]
        $f : \mathcal{P}(Q)$ is a \textbf{strategy}.
    \end{definition}

    Combining a Delivered predicate and a strategy results in runs capturing
    the behavior of processes for the corresponding Delivered collections
    when following the strategy. In these runs, processes can change round only when
    allowed by the strategy, and must also do so if the strategy allows
    it continuously.

    \begin{definition}[Runs Generated by a Strategy]
      \label{def:runstrategy}
        Let $f$ be a strategy and $t$ a run.
        $t$ is a \textbf{run generated by} $f \triangleq$
        $t$ satisfies the following:
        \begin{itemize}
            \item \textbf{(Next only if allowed)}
                $\forall i \in [0,size(t)), \forall j \in \Pi:
                (t[i].transition = next_j  \implies t[i].state.j \in f)$
            \item \textbf{(Infinite fairness of next)}
                If $t$ is infinite, then $\forall j \in \Pi:
                |\{i \in \mathbb{N} \mid t[i].transition = next_j\}| < \aleph_0
                \implies |\{i \in \mathbb{N} \mid t[i].state.j \notin f\}|
                = \aleph_0$
            \item \textbf{(Finite fairness of next)}
                If $t$ is finite, then $\forall j \in \Pi:
                t[size(t)-1].state.j \notin f$.
        \end{itemize}

        For a Delivered predicate \textit{PDel}, we note $runs_f(\textit{PDel})
        = \{\textit{t a run} \mid \textit{t generated by f} \land
        t \in runs(\textit{PDel}) \}$.
    \end{definition}

    From that point, it is clear that a well-behaved strategy is such
    that all its runs are valid.

    \begin{definition}[Valid Strategy]
        Let \textit{PDel} a Delivered predicate and $f$ a strategy.
        $f$ is a \textbf{valid strategy} for \textit{PDel}
        $\triangleq \forall t \in runs_f(PDel): t$ is a valid run.
    \end{definition}

    Validity guarantees an infinite number of complete
    rounds for every run of the strategy. This ensures that a run defines a Heard-Of
    collection, as we see next.

    \subsection{From Delivered Collections to Heard-Of Collections}

    Recall that the difference between a Heard-Of and a Delivered collection is that
    the latter takes into account all delivered messages, while the former
    only considers messages from a round if they were received
    before or during the corresponding round of the receiver.

    If a run is valid, then all processes have infinitely many rounds, and thus it
    defines a Heard-Of collection through its behavior.

    \begin{definition}[Heard-Of Collection of a Valid Run]
        Let $t$ a valid run.
        Then, $CHO_t$ is the \textbf{Heard-Of collection of} $t \triangleq\\
        \forall r \in \mathbb{N}^*, \forall j \in \Pi : CHO_t(r,j) =
        \left\{ k \in \Pi ~\middle|~ \exists i \in \mathbb{N}:
        \left(
        \begin{array}{ll}
            & t[i].state.j.round = r\\
            \land & t[i+1].j.state.round = r+1\\
            \land & \langle r,k \rangle \in t[i].state.j.received\\
        \end{array}
        \right)
        \right\}$
    \end{definition}

    It is useful to go the other way, and extract from a Heard-Of
    collection some canonical valid run generating it. Our choice is
    a run where processes change round in lockstep, every message
    from a Heard-Of set is delivered in the round where it was sent, and
    every late message is delivered in the round
    following the one where it was sent.

    \begin{definition}[Standard Run of a Heard-Of collection]
        Let \textit{cho} be a Heard-Of collection.
        For $r > 0$, let $\textit{onTimeMes}_r$ be a permutation of
         $\{deliver(r,k,j) \mid
        k, j \in \Pi \land k \in cho(r,j)\}$, $\textit{lateMes}_r$
        a permutation of
        $\{deliver(r-1,k,j) \mid k, j \in \Pi \land k \notin cho(r-1,j)\}$,
        and
        $\textit{nexts}$ be a permutation of $\{next_j \mid j \in \Pi\}$.
        Then, the run starting at the initial state and with
        transitions defined by the word
        $\prod\limits_{r > 0} (\textit{lateMes}_r . \textit{onTimeMes}_r .
        \textit{nexts})$ is a \textbf{standard run} of \textit{cho}.
    \end{definition}

    This canonical run is a run of any Delivered predicate containing the
    collection where every message is delivered. This collection captures
    the case where no failure occurs: every process always broadcasts and no
    message is lost. Having this collection in a Delivered predicate
    ensures that although faults might happen, they are not forced to do so.

    \begin{lemma}[Standard Run is a Run for Total Collection]
        \label{trivial}
        Let \textit{cho} be a Heard-Of collection and let
        \textit{PDel} be a Delivered
        predicate containing the total collection \textit{CDel}$_{total}$
        defined by $\forall r > 0, \forall j \in \Pi:
        \textit{CDel}_{\textit{total}}(r,j) = \Pi$. Then,
        a standard run $t$ of \textit{cho} is a run of \textit{PDel}.
    \end{lemma}

    \begin{proof}
        First, $t$ is a run since it satisfies the four constraints
        defining a run:
        \begin{itemize}
            \item The first state is the initial state by definition.
            \item We only gave the transitions, and the states
                changes as mandated by the transition function.
            \item Every message from round $r$ is delivered after
                either $r$ or $r+1$ \textit{next} transitions for
                the sender, which ensures it is delivered after
                being sent.
            \item Every sent message is delivered either during the
                round it was sent or during the next one, and thus
                delivered only once.
        \end{itemize}

        Furthermore, one message from each process is eventually delivered to
        everyone for each round in $t$,
        which means that $t$ is a run of the total Delivered
        collection, and thus a run of \textit{PDel}.
    \end{proof}

    \begin{definition}[Heard-Of Predicate Generated by Strategy]
        If $f$ is a valid strategy for \textit{PDel},
        we write $PHO_f(\textit{PDel})$ for the Heard-Of collections of the
        runs generated by $f$ for \textit{PDel}: $PHO_f(\textit{PDel})
        \triangleq \{ CHO_t \mid t \in runs_f(\textit{PDel}) \}$.
    \end{definition}

    Every valid strategy generates a Heard-Of predicate from the
    Delivered predicate. We now have a way to go from a Delivered predicate
    to a Heard-Of one: design a valid strategy for the former that generates
    the latter. But we still have not answered the original question: among
    all the Heard-Of predicates one can generate
    from a given Delivered predicate, which one should we consider as the
    characterization of the Delivered predicate?

    First, remark that all Delivered collections
    can be generated as Heard-Of collections by a valid strategy:
    simply deliver all messages from a round before changing the round
    of a process -- the change of round must eventually happen by validity
    of the strategy. Thus, every Heard-Of predicate generated from a Delivered
    one is an overapproximation of the latter: for any strategy $f$, $PDel \subseteq PHO_f(PDel)$. But we want to receive as
    many messages as possible on time, that is to be as close as possible
    to the original Delivered predicate. The characterizing Heard-Of predicate
    is thus the smallest such overapproximation, if it exists.

    We formalize this intuition by defining a partial order on valid strategies
    for a Delivered predicate capturing the implication of
    the generated Heard-Of predicates. One strategy dominates another if the
    Heard-Of set it generates is included in the one generated by the other.
    Dominating strategies are then the greatest elements for this order. By
    definition of domination, all dominating strategies generate the same
    dominating Heard-Of predicate, which characterizes the Delivered predicate.

    \begin{definition}[Domination Order, Dominating Strategy and Dominating Predicate]
        Let \textit{PDel} be a Delivered predicate and let
        $f$ and $f'$ be two valid strategies for \textit{PDel}.
        Then, $f$~\textbf{dominates} $f'$ for \textit{PDel},
        written $f' \prec_{\textit{PDel}} f \triangleq
        PHO_{f'}(\textit{PDel}) \supseteq PHO_f(\textit{PDel})$.

        A greatest element for $\prec_{\textit{PDel}}$ is called
        a \textbf{dominating strategy} for \textit{PDel}. Given
        such a strategy $f$, the \textbf{dominating predicate}
        for \textit{PDel} is then $PHO_f(\textit{PDel})$.
    \end{definition}

\section{A Complete Example: At Most $F$ Crashes}
\label{sec:formalexample}

    To provide a more concrete intuition, we turn to an example:
    the message-passing model with
    asynchronous and reliable communication, and at most $F$ permanent
    crashes. To find the corresponding Delivered predicate, we characterize
    which messages are delivered in a round execution of this model:
    all the messages sent by a process before it crashes are delivered; it sends
    no message after; at the round where it crashes, there may
    be an incomplete broadcast if the process crashes in the middle of it.
    Additionally, at most $F$ processes can crash.

    \begin{definition}[\textit{PDel}$^{F}$]
        The Delivered predicate \textit{PDel}$^{F}$ for the asynchronous model
        with reliable communication and at most $F$ permanent crashes
        $\triangleq\\
        \left\{ \textit{CDel} \in (\mathbb{N}^* \times \Pi) \mapsto \mathcal{P}(\Pi)
        ~\middle|~
            \forall r > 0, \forall j \in \Pi:
            \begin{array}{ll}
                & |CDel(r,j)| \geq n-F\\
                \land & CDel(r+1,j) \subseteq K_{\textit{CDel}}(r)\\
            \end{array}
        \right\}$.
    \end{definition}

    The folklore strategy for this model is to wait for at least $n-F$ messages
    before allowing the change of round.

    \begin{definition}[waiting for $n-F$ messages]
      The strategy to wait for $n-F$ messages is:
        $f_{n-F} \triangleq
        \{ q \in Q \mid
        |\{k \in \Pi \mid \langle q.round, k \rangle \in q.received\}| \geq n-F \}$
    \end{definition}

    To see why this strategy is used in the literature, simply remark that
    at least $n-F$ messages must be delivered to each process at each round.
    Thus, waiting for that many messages ensures that no process is ever
    blocked. Rephrased with the concepts introduced above,
    $f_{n-F}$ is a valid strategy for \textit{PDel}$^F$.

    \begin{lemma}[Validity of $f_{n-F}$]
        \label{nfValid}
        $f_{n-F}$ is valid for \textit{PDel}$^{F}$.
    \end{lemma}

    \begin{proof}
        We proceed by contradiction: \textbf{Assume} $f_{n-F}$ is invalid
        for \textit{PDel}$^F$. Thus, there exists
        $t \in runs_{f_{n-F}}(\textit{PDel}^F)$ invalid.
        Because $t$ is infinite, the problem
        is either the infinite fairness of \textit{next} or a \textit{next} done
        when $f_{n-F}$ does not allow it. Each \textit{next} transition being played
        an infinite number of times, we conclude that some \textit{next} is
        played while $f_{n-F}$ does not allow it.
        Let $r$ be the smallest round where it happens
        and $j$ be a process blocked at $r$ in $t$.
        Let also \textit{CDel}$_t$ be a Delivered collection of \textit{PDel}$^F$
        such that $t \in runs(\textit{CDel}_t)$.

        We know by definition of \textit{PDel}$^F$ that
        $|\textit{CDel}_t(r,j)| \geq n-F$. The minimality of $r$ and
        the fact that $t \in runs(\textit{CDel})$ then ensure
        that all messages in this Delivered set are delivered at some point
        in $t$. By definition of $f_{n-F}$, the
        transition \textit{next}$_j$ is then available from this
        point on. This \textbf{contradicts} the fact that $j$ cannot
        change round at this point in $t$.
    \end{proof}

    The Heard-Of predicate generated by $f_{n-F}$ was first given
    by Charron-Bost and Schiper~\cite{CharronBostHO} as a characterization
    the asynchronous model with reliable communication and
    at most $F$ crashes. The intuition behind it is that even in the absence of
    crashes, we can make all processes change round by delivering any set
    of at least $n-F$ messages to them.

    \begin{theorem}[Heard-Of Characterization of $f_{n-F}$]
        \label{n-fCharac}~\\
        $PHO_{f_{n-F}}(\textit{PDel}^F) =
        \{ \textit{cho} \in (\mathbb{N}^* \times \Pi) \mapsto \mathcal{P}(\Pi)
        \mid \forall r \in \mathbb{N}^*, \forall j \in \Pi:
        |\textit{cho}(r,j)| \geq n-F \}$.
    \end{theorem}

    \begin{proof}
        First, we show $\subseteq$.
        Let $\textit{cho} \in PHO_{f_{n-F}}(\textit{PDel}^F)$
        and $t \in runs_{f_{n-F}}(\textit{PDel}^F)$ a run of $f_{n-F}$
        generating \textit{cho}.
        By definition of the runs of $f_{n-F}$, processes change
        round only when they received at least $n-F$ messages from the current
        round, which implies that
        $\forall r \in \mathbb{N}^*, \forall j \in \Pi : |cho(r,j)| \geq n-F$.

        Then, we show $\supseteq$.
        Let \textit{cho} a Heard-Of collection over $\Pi$ such that
        $\forall r \in \mathbb{N}, \forall j \in \Pi: |\textit{cho}(r,j)| \geq n-F$.
        Let $t$ be a standard run of \textit{cho}; since \textit{PDel}$^F$
        contains the total collection, $t$ is a run of \textit{PDel}$^F$
        by Lemma~\ref{trivial}.
        To prove this is also a run of $f_{n-F}$, we proceed by contradiction:
        \textbf{Assume} it is not: because $t$ is infinite, the problem
        is either the infinite fairness of \textit{next} or a \textit{next} done
        when $f_{n-F}$ does not allow it. Each  \textit{next} transition being played
        an infinite number of times in a standard run,
        the only possibility left is the second one: some \textit{next}
        transition in $t$ is done while
        $f_{n-F}$ does not allow the corresponding process to change round.
        Let $r$ be the smallest round where this happens, and
        $j$ one of the concerned processes at round $r$.
        By definition of $t$ as a standard run,
        $j$ received all messages from \textit{cho}$(r,j)$
        before the problematic \textit{next}. And $|\textit{cho}(r,j)| \geq n-F$
        by hypothesis. By definition of $f_{n-F}$, the
        transition \textit{next}$_j$ is then available from this
        point on. This \textbf{contradicts} the fact that $j$ cannot
        change round at this point.
        We conclude that $\textit{cho} \in PHO_{f_{n-F}}(\textit{PDel}^F)$.
    \end{proof}

    Finally, we want to vindicate the folklore intuition about this strategy:
    that it is optimal in some sense. Intuitively, waiting for more than $n-F$
    messages per round means risking waiting forever, and waiting for less
    is wasteful. Our domination order captures this concept of optimality:
    we show that $f_{n-F}$ is indeed a dominating strategy for \textit{PDel}$^F$.
    Therefore, \textit{PHO}$_{f_{n-F}}(\textit{PDel}^F)$ is the
    dominating predicate for \textit{PDel}$^F$.

    To do so, we introduce another canonical run, this time for the
    combination of a Delivered collection and a strategy.
    This run consists
    in successive iterations where all messages sent and not yet delivered
    are delivered, and then all processes which are allowed to change round
    by the strategy do. This captures the run where every message is delivered
    as early as possible after being sent.

    \begin{definition}[Earliest Run of Strategy for Delivered Collection]
        Let \textit{CDel} be a Delivered collection and $f$ be a strategy.
        We now define
        for $r > 0$ the sets $\textit{dels}_r$ and $\textit{nexts}_r$, as well as
        the states $qDels_r$ and $qNexts_r$. The two sets are respectively
        the set of deliveries at iteration $r$ and the set of \textit{nexts}
        at iteration $r$. As for the states, they respectively capture
        the state of the system just before the first delivery of the $r$-ith
        iteration and just before the first \textit{next}
        of this iteration. We define them all by the following recurrence:
        \begin{itemize}
            \item \textbf{State at iteration $r$ after nexts (of $r$)}
                We have $qDels_1 = \langle 1, \emptyset \rangle^n$ and\\
                $\forall r > 1:
                \left(
                \begin{array}{ll}
                    \forall j \in \Pi: & \\

                    \textit{if next}_j \in \textit{nexts}_r &
                    \left(
                    \begin{array}{l}
                        \textit{qDels}_r.j = \textit{qNexts}_{r-1}.j
                        \textit{ Except }\\
                        \quad \textit{qDels}_r.j.round =
                        \textit{qNexts}_{r-1}.j.round + 1
                    \end{array}
                    \right)\\

                    \textit{otherwise} & \textit{qDels}_r.j =
                    \textit{qNexts}_{r-1}.j\\
                \end{array}
                \right)$
            \item \textbf{Delivers at iteration $r$}\\
                We have
                $\textit{dels}_1 = \{\textit{deliver}(1,k,j) \mid
                k \in \textit{CDel}(1,j)\}$ and
                $\forall r > 1: \textit{dels}_r = \{\textit{deliver}(r',k,j)
                \mid  r' = \textit{qDels}_r.j.round \land r' \neq
                \textit{qNexts}_{r-1}.j.round \land k \in \textit{CDel}(r',j)\}$.
            \item \textbf{State at iteration $r$ after deliveries}\\
                $\forall r > 0:
                \left(
                \begin{array}{l}
                    \forall j \in \Pi: \textit{qNexts}_r.j =
                    \textit{qDels}_r.j
                    \textit{ Except }\\
                    \quad\textit{qNexts}_r.j.received =
                    \textit{qDels}_r.j.received \cup
                    \{\langle r, k \rangle \mid
                    \langle r, k, j \rangle \in \textit{dels}_r\}\\
                \end{array}
                \right)$
            \item \textbf{Nexts allowed at iteration $r$}\\
                $\forall r > 0: \textit{nexts}_r = \{\textit{next}_j \mid
                \textit{qNexts}_r.j \in f\}$.
        \end{itemize}

        For all $r > 0$, let $Wdels_r$ be a permutation of
        $\textit{dels}_r$, and $\textit{Wnexts}_r$ a permutation of
        $\textit{nexts}_r$.

        Then, the run starting at the initial state and
        with the following transitions:
        $\sum\limits_{r > 0} (\textit{Wdels}_r . \textit{Wnexts}_r)$ -- with
        an \textit{end} at the end if it is a finite run -- is
        an \textbf{earliest run} of $f$ for \textit{CDel}.
    \end{definition}

    \begin{lemmabis}[Earliest Run is a Run]
        \label{trivialBis}
        Let \textit{CDel} be a Delivered collection and $f$ be a strategy.
        Then, an earliest run of $f$ for \textit{CDel} is a run of $f$.
    \end{lemmabis}

  \begin{proof}
        First, $t$ is a run since it satisfies the four constraints
        defining a run:
        \begin{itemize}
            \item The first state is the initial state by definition.
            \item Transitions change state as required by definition.
            \item Every message from round $r$ is delivered while the
                emitter is in round $r$,
                which ensures it is delivered after being sent.
            \item Every sent message is delivered during the
                round it was sent, and thus delivered only once.
        \end{itemize}

        It is also a run of $f$ because it satisfies the three constraints:
        \begin{itemize}
            \item By definition, processes only change round when allowed
                by $f$.
            \item If $t$ is infinite, then a process which has only finitely
                many changes of round is blocked at some round forever;
                this means that there are infinitely many iterations where
                $f$ does not continuously enable the change of round.
            \item If $t$ is finite on the other hand, this means that for
                all iterations from a point on, no process is allowed to
                change round, and thus that in their last state all
                processes are forbidden to change round by $f$.
        \end{itemize}
        We therefore deduce $t$ is a run of $f$.
    \end{proof}

    By combining standard and earliest runs, we show that any valid
    strategy for \textit{PDel}$^F$ is dominated by $f_{n-F}$  and thus
    that \textit{PHO}$_{f_{n-F}}(\textit{PDel}^F)$ is the dominating
    predicate for \textit{PDel}$^F$.

    \begin{theorem}[$f_{n-F}$ Dominates \textit{PDel}$^F$]
        \label{fnfDom}
        $f_{n-F}$ dominates \textit{PDel}$^F$.
    \end{theorem}

    \begin{proof}
        Let $f$ be a valid strategy for \textit{PDel}$^F$; we now prove
        that $f \prec_{\textit{PDel}^F} f_{n-F}$, that is
        $PHO_{f_{n-F}}(\textit{PDel}^F) \subseteq PHO_f(\textit{PDel}^F)$.
        Let $\textit{cho} \in PHO_{f_{n-F}}(\textit{PDel}^F)$, and
        let $t$ be a standard run of \textit{cho}. Since $\textit{PDel}^F$
        contains the total collection, $t$ is a run of \textit{PDel}$^F$
        by Lemma~\ref{trivial}. We only need to prove that it is also a run
        of $f$ to conclude.

        We do so by contradiction.
        \textbf{Assume} $t$ is not a run of $f$:
        because $t$ is infinite, the problem
        is either the infinite fairness of \textit{next} or a \textit{next} done by a process $j$
        when $f$ does not allow it. Each  \textit{next} transition being played
        an infinite number of times in a standard run,
        the only possibility left is the second one.
        At the point of the forbidden \textit{next}, by definition of a
        standard run, $j$ has received
        every message from previous rounds, and all messages from
        $\textit{cho}(r,j)$. By application of Theorem~\ref{n-fCharac}
        and \textit{cho} being in \textit{PHO}$_{f_{n-F}}(\textit{PDel}^F)$,
        $\textit{cho}(r,j)$ contains at least $n-F$ processes.

        Let \textit{CDel}$_{\textit{block}}$ be the Delivered collection
        where all processes from which $j$ did not receive a message at
        the problematic \textit{next} in $t$ stop sending messages from
        this round on:
        $\forall r' > 0, \forall k \in \Pi:
        \textit{CDel}_{\textit{block}}(r',k) =
        \left\{
        \begin{array}{ll}
            \Pi & \textit{if } r' < r\\
            cho(r,j) & \textit{otherwise}\\
        \end{array}
        \right.$

        This is a Delivered collection of \textit{PDel}$^F$:
        processes that stop sending messages never do again,
        and at most $F$ processes do so because $\textit{cho}(r,j)$
        contains at least $n-F$ processes.

        Let $t_{\textit{block}}$ be an earliest run of $f$ for
        \textit{CDel}$_{\textit{block}}$.
        This is a run of $f$, by Lemma~\ref{trivialBis}.
        We then have two possibilities.
        \begin{itemize}
            \item During one of the first $r-1$ iterations of
                $t_{\textit{block}}$, there is some process which cannot change
                round. Let $r'$ be the smallest iteration where it happens,
                and $k$ be a process unable to change round at this
                iteration.
                By minimality of $r'$, all processes arrive at round $r'$,
                and by symmetry of $\textit{CDel}_{\textit{block}}$ they all
                receive the same messages as $k$. Thus, all processes
                are blocked at round $r'$, there are no more next or deliveries,
                and $t_{block}$ is therefore invalid.
            \item For the first $r-1$ iterations, all processes change round.
                Thus, every one arrives at round~$r$. By definition of
                an earliest run, all messages from the round are delivered
                before any \textit{next}. The symmetry of
                $\textit{CDel}_{\textit{block}}$ also ensures that every process
                received the same messages, that is all messages from
                round $< r$ and all messages from \textit{cho}$(r,j)$.
                These are exactly the messages received by $j$ in $t$ at
                round $r$.
                But by hypothesis, $j$ is blocked in this state in $t$.
                We thus deduce that all processes are blocked at round $r$
                in $t_{block}$, and thus that it is an invalid run.
        \end{itemize}
        Either way, we deduce that $f$ is invalid, which is a
        \textbf{contradiction}.
    \end{proof}

    This means that when confronted with a model captured by
    \textit{PDel}$^F$, there is no point in remembering messages from past
    rounds -- and messages from future rounds are simply buffered.
    Intuitively, messages from past rounds are of no use in detecting
    crashes in the current round. As for messages from
    future rounds, they could serve to detect that a process has not crashed when
    sending its messages from the current round. This does not
    alter the Heard-Of predicate because nothing forces messages from future
    rounds to be delivered early, and thus there is no way to systematically use
    the information from future rounds.

\section{Carefree Strategies}
\label{sec:memless}

    We now turn to more general
    results about Delivered predicates and strategies. We focus first
    on a restricted form of strategies, the carefree ones: they depend only
    on the received messages from the current round. For example, $f_{n-F}$ is
    a carefree strategy. These are quite simple strategies,
    yet they can be dominating, as shown for $f_{n-F}$ and \textit{PDel}$^F$.

    \subsection{Definition and Expressiveness Results}

    \begin{definition}[Carefree Strategy]
      \label{carefree}
        Let $f$ be a strategy and, $\forall q \in Q$, let $cfree(q) = \{
            k \in \Pi \mid \langle q.round, k \rangle \in q.received \}$.
        $f$ is a \textbf{carefree strategy}
        $\triangleq \forall q, q' \in Q:
        cfree(q) = cfree(q') \implies (q \in f \iff q' \in f)$.

        For $f$ a carefree strategy, let $\textit{Nexts}_{f} \triangleq \{cfree(q) \mid q \in f \}$.
        It uniquely defines $f$.
    \end{definition}

    Thus a carefree strategy can be defined by a set of sets
    of processes: receiving a message from all processes in any of those
    set makes the strategy authorize the change of round. This gives us a simple
    necessary condition on such a strategy to be valid: its \textit{Nexts} set
    must contains all Delivered set from the corresponding Delivered predicate.
    If it does not, then an earliest run of any collection containing a
    Delivered set not in the \textit{Nexts} would be invalid.

    This simple necessary condition also proves sufficient.

    \begin{lemma}[Validity of Carefree]
        \label{cfreeValid}
        Let \textit{PDel} be a Delivered predicate and
        $f$ a carefree strategy.
        Then, $f$ is valid for \textit{PDel} $\iff
        \forall \textit{CDel} \in \textit{PDel}, \forall r > 0, \forall j \in \Pi:
        \textit{CDel}(r,j) \in \textit{Nexts}_f$.
    \end{lemma}

    \begin{proof}
        $(\Rightarrow)$ Let $f$ be valid for \textit{PDel}.
        We show by contradiction that it satisfies the right-hand side of the
        above equivalence.
        \textbf{Assume} there is
        $\textit{CDel} \in \textit{PDel}, r > 0$ and $j \in \Pi$ such that
        \textit{CDel}$(r,j) \notin \textit{Nexts}_f$.
        Then, let $t$ be an earliest run of $f$ for
        \textit{CDel}.
        This is a run of $f$ by Lemma~\ref{trivialBis}.

        The sought contradiction is reached by proving that $t$ is invalid.
        To do so, we split according to two cases.
        \begin{itemize}
            \item During one of the first $r-1$ iterations of
                $t$, there is some process which cannot change
                round. Let $r'$ be the smallest iteration where it happens,
                and $k$ be a process unable to change round at the $r'$-ith
                iteration.

                By minimality of $r'$, all processes arrive at round $r'$ in
                $t$; by definition of an earliest run, all messages for $k$
                from round $r'$ are delivered before the \textit{next}
                for the iteration.
                Let $q$ the local state of $k$ at the
                first \textit{next} in the $r'$-ith iteration, and
                let $q'$ be any local state of $k$ afterward. 
                The above tells us that as long as $q'.round = q.round$,
                we have $cfree(q) = cfree( q')$ and thus $q' \notin f$.
                Therefore, $k$ can never change round while at round $r'$.

                We conclude that $t$ is invalid.
            \item For the first $r-1$ iterations, all processes change round.
                Thus every one arrives at round $r$ in the $r-1$-ith iteration.
                By definition of
                an earliest run, all messages from the round are delivered
                before any \textit{next}. By hypothesis, $j$ cannot
                change round because its Delivered set is not in
                $\textit{Nexts}_f$.
                Let $q$ the local state of $j$ at the
                first \textit{next} in the $r$-ith iteration, and
                let $q'$ be any local state of $j$ afterward. 
                The above tells us that as long as $q'.round = q.round$,
                we have $cfree(q) = cfree(q') = \textit{CDel}(q.round,j)$
                and thus $q' \notin f$. Therefore, $j$ can never
                change round while at round $r$.

                Here too, $t$ is invalid.
        \end{itemize}

        Either way, we reach a \textbf{contradiction} with the validity of $f$.

        $(\Leftarrow)$ Let \textit{PDel} and $f$ such that
        $\forall \textit{CDel} \in \textit{PDel},
        \forall r > 0, \forall j \in \Pi: \textit{CDel}(r,j) \in \textit{Nexts}_f$.
        We show by contradiction that $f$ is valid.

        \textbf{Assume} the contrary: there is some
        $t \in \textit{runs}_f(\textit{PDel})$ which is invalid. Thus, there
        are some process blocked at a round forever in $t$. Let $r$ be the
        smallest such round, and $j$ be a process blocked at round $r$ in $t$.
        By minimality of $r$, all processes arrive at round $r$.
        By definition of a run of \textit{PDel}, there is a
        $\textit{CDel} \in \textit{PDel}$ such that $t$ is a run of \textit{CDel}.
        Thus, eventually all messages from $\textit{CDel}(r,j)$ are delivered.

        From this point on, $f$ allows $j$ to change round by definition, and
        the fairness of \textit{next} imposes that $j$ does at some point.
        We conclude that $j$ is not blocked at round $r$ in $t$,
        which \textbf{contradicts} the hypothesis.
    \end{proof}

    Carefree strategies are elegant, but they also have some drawbacks:
    mainly that the Heard-Of predicates they can implement are quite basic.
    Precisely, when the Delivered predicate contains the total collection,
    they implement predicates where the Heard-Of collections are
    all possible combinations of Delivered sets from the original Delivered
    predicate.

    \begin{theorem}[Heard-Of Predicates of Carefree Strategy]
        \label{cfreeHO}
        Let \textit{PDel} be a Delivered predicate containing the total
        collection, and let $f$ be a valid carefree strategy for \textit{PDel}.
        Then, $\forall \textit{cho}$ a Heard-Of collection for $\Pi:
        \textit{cho} \in CHO_f(\textit{PDel}) \iff
        \forall r > 0, \forall j \in \Pi: \textit{cho}(r,j) \in
        \textit{Nexts}_f$.
    \end{theorem}

    \begin{proof}
        $(\Rightarrow)$ This direction follows from
        the definition of a carefree strategy: it allows changing round
        only when the messages received from the current round form a set
        in its \textit{Nexts}.

        $(\Leftarrow)$ Let \textit{cho} be a Heard-Of collection for $\Pi$
        such that $\forall r > 0, \forall j \in \Pi: \textit{cho}(r,j) \in
        \textit{Nexts}_f$. Let $t$ be a standard run of \textit{cho}. It is
        a run by Lemma~\ref{trivial}. It is also a run of $f$ because
        at each round, processes receive messages from a set in $\textit{Nexts}_f$
        and are thus allowed by $f$ to change round.
        We conclude that $\textit{cho} \in \textit{PHO}_f(\textit{PDel})$.
    \end{proof}

    Finally, carefree strategies always have a dominating element
    for a given Delivered predicate.
        This is because the strategy waiting
        for exactly the Delivered sets of the predicate waits for more
        messages than any other valid carefree strategy, by Lemma~\ref{cfreeValid}.

    \begin{theorem}[Always a Dominating Carefree Strategy]
        \label{domCfree}
        Let \textit{PDel} be a Delivered predicate and
        let $f_{cfDom}$ be the carefree strategy with \textit{Nexts}$_f =
        \{ \textit{CDel}(r,j) \mid \textit{CDel} \in \textit{PDel}
                                    \land r > 0 \land j \in \Pi\}$.
        $f_{cfDom}$ dominates all carefree strategies for \textit{PDel}.
    \end{theorem}

    \begin{proof}
        First, $f_{cfDom}$ is valid for \textit{PDel} by application of
        Lemma~\ref{cfreeValid}.

        As for domination, we also deduce from Lemma~\ref{cfreeValid} that
        $\forall f$ a valid carefree strategy for \textit{PDel},
        $\textit{Nexts}_{f_{cfDom}} \subseteq \textit{Nexts}_f$ and thus
        $f_{cfDom} \subseteq f$. Therefore,
        $\forall q \in Q:
        q \in f_{cfDom} \implies q \in f$. This gives us
        $runs_{f_{cfDom}}(\textit{PDel}) \subseteq runs_f(\textit{PDel})$,
        and we conclude $PHO_{f_{cfDom}}(\textit{PDel})
        \subseteq PHO_f(\textit{PDel})$.
        Therefore, $f_{cfDom}$ dominates all valid carefree strategies
        for \textit{PDel}.
    \end{proof}

    In the case where \textit{PDel} allows the delivery of all messages,
    that is contains the total collection, there is only one carefree strategy
    which dominates all carefree strategies.

    \begin{theorem}[With Total Collection, Unique Dominating Carefree Strategy]
        \label{uniqDomCfree}
        Let \textit{PDel} be a Delivered predicate containing
        the total collection. Let $f_{cfDom}$ be the carefree strategy
        with \textit{Nexts}$_f =
        \{ \textit{CDel}(r,j) \mid \textit{CDel} \in \textit{PDel}
                                    \land r > 0 \land j \in \Pi\}$.
        $f_{cfDom}$ is the unique carefree strategy which dominates
        all carefree strategies for \textit{PDel}.
    \end{theorem}

    \begin{proof}
        First, $f_{cfDom}$ dominates all carefree strategies for \textit{PDel}
        by Theorem~\ref{domCfree}.
        To show uniqueness, let $f$ be a valid carefree strategy different
        from $f_{cfDom}$; thus $\textit{Nexts}_{f} \neq
        \textit{Nexts}_{f_{cfDom}}$. By Lemma~\ref{cfreeValid}, we also have
        $\textit{Nexts}_{f_{cfDom}} \subseteq \textit{Nexts}_f$. Therefore,
        $\textit{Nexts}_{f_{cfDom}} \subsetneq \textit{Nexts}_f$.

        Let $D \in \textit{Nexts}_f \setminus \textit{Nexts}_{f_{cfDom}}$
        and let $\textit{cho} \in
        (\mathbb{N}^* \times \Pi) \mapsto \mathcal{P}(\Pi)$ the Heard-Of
        collection such that $\forall r > 0, \forall j \in \Pi:
        \textit{cho}(r,j) = D$. By application of Theorem~\ref{cfreeHO},
        this is a Heard-Of collection generated by $f$ but not
        by $f_{cfDom}$. Therefore,
        $\textit{PHO}_{f}(\textit{PDel}) \not\subset
        \textit{PHO}_{f_{cfDom}}(\textit{PDel})$,
        and $f'$ does not dominate $f_{cfDom}$.
    \end{proof}

    \subsection{When Carefree is Enough}

    Finally, the value of carefree strategy depends on which Delivered
    predicates have such a dominating strategy. We already know that
    \textit{PDel}$^F$ does; we now extend this result to a class
    of Delivered predicates called Round-symmetric.
    This condition captures the fact that given any Delivered set $D$,
    one can build, for any $r > 0$, a Delivered collection where processes
    receive all messages up to round $r$, and then they share $D$ as their
    Delivered set in round $r$. As a limit case, the predicate also contains
    the total collection.

    \begin{definition}[Round-Symmetric Delivered Predicate]
        Let \textit{PDel} be a Delivered Predicate.
        \textit{PDel} is \textbf{round-symmetric} $\triangleq$
        \begin{itemize}
            \item \textbf{(Total collection)}
                \textit{PDel} contains the total collection: $\textit{CDel}_{total} \in \textit{PDel}$ where $\textit{CDel}_{total}$ is defined by $\forall r > 0, \forall j \in \Pi: \textit{CDel}_{\textit{total}}(r,j) = \Pi$.
            \item \textbf{(Symmetry up to a round)} $\forall D \in
                \{ \textit{CDel}(r,j) \mid \textit{CDel} \in \textit{PDel}
                \land r > 0 \land j \in \Pi\},\\
                \forall r > 0,
                \exists \textit{CDel} \in \textit{PDel}, \forall j \in \Pi:
                    (\forall r' < r: \textit{CDel}(r',j) = \Pi \land \textit{CDel}(r,j) = D)
                $
        \end{itemize}
    \end{definition}

    What round-symmetry captures is what makes \textit{PDel}$^F$
    be dominated by a carefree strategy: the inherent symmetry of these Delivered
    collections allows us to block processes with exactly the same received messages.
    This allows us to show that any valid strategy should allow changing 
    round at this point, which is fundamental to any proof of domination.

    \begin{theorem}[Sufficient Condition of Carefree Domination]
        \label{suffCfreeDom}
        Let \textit{PDel} be a Round-symmetric Delivered predicate.
        Then, there is a carefree strategy which dominates \textit{PDel}.
    \end{theorem}

    \begin{proof}
        Let $f$ be a carefree strategy dominating all carefree strategies
        for \textit{PDel} -- it exists
        by Theorem~\ref{domCfree} -- and let $f'$ be a valid strategy
        for \textit{PDel}.
        We now prove that $f' \prec_{\textit{PDel}^F} f$, that is
        $PHO_f(\textit{PDel}) \subseteq PHO_{f'}(\textit{PDel})$.
        Let also $\textit{cho} \in PHO_f(\textit{PDel})$
        and $t$ be a standard run of \textit{cho}.

        By Lemma~\ref{trivial}, $t$ is a run; we now prove by contradiction
        it is also a run of $f'$. \textbf{Assume} it is not:
        because $t$ is infinite, the problem
        is either the infinite fairness of \textit{next} or a \textit{next} done
        when $f'$ does not allow it. Each  \textit{next} transition being played
        an infinite number of times in a standard run,
        the only possibility left is the second one: some \textit{next}
        transition in $t$ is done while
        $f'$ does not allow the corresponding process to change round.
        There thus exists a smallest $r$ such that some process $j$ is not allowed
        by $f'$ to change round when \textit{next}$_j$ is played at round $r$
        in $t$.
        Because \textit{PDel} contains the total collection, there is only
        one carefree strategy  dominating all carefree strategies for
        \textit{PDel} by Theorem~\ref{uniqDomCfree}. This is the strategy
        from Theorem~\ref{domCfree}. Thus, the application of
        Theorem~\ref{cfreeHO} yields that $\textit{cho}(r,j) \in
        \textit{Nexts}_f =
        \{ \textit{CDel}(r,j) \mid \textit{CDel} \in \textit{PDel}
                                   \land r > 0 \land j \in \Pi\}$.

        Next, the round-symmetry of \textit{PDel} allows us to build
        $\textit{CDel}_{block}
        \in \textit{PDel}$ such that
        $\forall r' < r, \forall k \in \Pi:
        \textit{CDel}_{block}(r',k) = \Pi$ and
        $\forall k \in \Pi:
        \textit{CDel}_{block}(r,k) = \textit{cho}(r,j)$.

        Finally, we build $t_{block}$, an earliest run of $f'$
        for \textit{CDel}$_{block}$.
        By Lemma~\ref{trivialBis}, we know $t_{block}$ is a run of $f'$.
        We then have two possibilities.
        \begin{itemize}
            \item During one of the first $r-1$ iterations of
                $t_{\textit{block}}$, there is some process which cannot change
                round. Let $r'$ be the smallest iteration where it happens,
                and $k$ be a process unable to change round at the $r'$-ith
                iteration.

                By minimality of $r'$, all processes arrive at round $r'$,
                and by symmetry of $\textit{CDel}_{\textit{block}}$ they all
                receive the same messages as $k$. Thus, all processes
                are blocked at round $r'$, there are no more next or deliveries,
                and $t_{block}$ is therefore invalid.
            \item For the first $r-1$ iterations, all processes change round.
                Thus, every one arrives at round $r$. By definition of
                an earliest run, all messages from the round are delivered
                before any \textit{next}. The symmetry of
                $\textit{CDel}_{\textit{block}}$ also ensures that every process
                received the same messages, that is all messages from
                round $< r$ and all messages from \textit{cho}$(r,j)$.
                These are exactly the messages received by $j$ in $t$ at
                round $r$.
                But by hypothesis, $j$ is blocked in this state in $t$.
                We thus deduce that all processes are blocked at round $r$
                in $t_{block}$, and thus that it is an invalid run.
        \end{itemize}

        Either way, we deduce that $f'$ is invalid, which is a
        \textbf{contradiction}.
    \end{proof}

    As another example of a Delivered predicate satisfying this condition,
    we study the model where at most $B$ broadcasts per round can fail:
    either all processes receive the message sent by a process at a round
    or none does; there are at most $B$ processes per round that can be in the
    latter case, where no one receives their message.
    The Delivered predicate states that each process receives the kernel
    of the round, and this kernel contains at least $n-B$ processes.

    \begin{definition}[\textit{PDel}$^B$]
        The Delivered predicate \textit{PDel}$^B$ corresponding to
        at most $B$ full broadcast failures is:\\
        $\{\textit{CDel} \in (\mathbb{N}^* \times \Pi) \mapsto \mathcal{P}(\Pi) \mid
        \forall r > 0, \forall j \in \Pi: \textit{CDel}(r,j) = K_{\textit{CDel}}(r)
        \land |K_{\textit{CDel}}(r)| \geq n-B\}$.
    \end{definition}

    The astute reader might have noticed that the carefree strategy dominating
    all carefree for this Delivered predicate is $f_{n-B}$ (which is $f_{n-F}$ with
    $F$ instantiated to $B$).

    \begin{lemma}[$f_{n-B}$ Carefree Dominates \textit{PDel}$^B$]
        $f_{n-B}$ dominates carefree strategies for \textit{PDel}$^B$.
    \end{lemma}

    \begin{proof}
        The definition of \textit{PDel}$^B$ gives us
        $\{\textit{CDel}(r,j) \mid \textit{CDel} \in \textit{PDel} \land r > 0
        \land j \in \Pi \}
        = \{ S \in \mathcal{P}(\Pi) \mid |S| \geq n-B \} = Nexts_{f_{n-B}}$,
        because the only constrained on the Delivered sets are that they must be
        the same for all processes at a round and that they must have at least
        $n-B$ processes.
        We therefore conclude from Lemma~\ref{cfreeValid} that $f_{n-B}$ is
        valid for \textit{PDel}$^B$ and from Theorem~\ref{domCfree} that
        it dominates carefree strategy for the same predicate.
    \end{proof}

    \begin{theorem}[$f_{n-B}$ Dominates \textit{PDel}$^B$]
        $f_{n-B}$ dominates \textit{PDel}$^B$.
    \end{theorem}

    \begin{proof}
        We only need to prove that \textit{PDel}$^B$ is round-symmetric;
        the application of Theorem~\ref{suffCfreeDom} will then
        yield that \textit{PDel}$^B$ is dominated by a carefree strategy,
        thus dominated by the carefree strategy that dominates all
        carefree strategies.
        \begin{itemize}
            \item \textit{PDel}$^B$ contains the total collection, because
                it is the collection where every round kernel is $\Pi$.
            \item Let $D \in
                \{ \textit{CDel}(r,j) \mid \textit{CDel} \in \textit{PDel}^B
                \land r > 0 \land j \in \Pi\}$. By definition of \textit{PDel}$^B$,
                $D$ contains at least $n-B$ processes.
                Let $r > 0$.

                We then build $\textit{CDel}:
                (\mathbb{N}^* \times \Pi) \mapsto \mathcal{P}(\Pi)$ such that
                $\forall j \in \Pi:
                \left(
                \begin{array}{l@{~}l}
                    & \forall r' < r: \textit{CDel}(r',j) = \Pi\\
                    \land & \forall r' \geq r: \textit{CDel}(r,j) = D
                \end{array}
                \right)$.\\
                At each round, the Delivered sets contains at least
                $n-B$ processes because $D$ does, and all processes
                share the same Delivered set. We conclude that $\textit{CDel}
                \in \textit{PDel}^B$, and thus that the symmetry up to
                a round is satisfied.
       \end{itemize}
    \end{proof}

    \begin{theorem}[Heard-Of Characterization of \textit{PDel}$^B$]
        The Heard-Of predicate implemented by $f_{n-B}$ in \textit{PDel}$^B$
        is $CHO_{f_{n-B}}(\textit{PDel}^B) =
        \{ \textit{cho} \in (\mathbb{N}^* \times \Pi) \mapsto \mathcal{P}(\Pi)
        \mid \forall r > 0, \forall j \in \Pi: |\textit{cho}(r,j)| \geq n-B \}$.
    \end{theorem}

    \begin{proof}
        We know that $\textit{Nexts}_{f_{n-B}}= \{ cfree(q) \mid
        q \in Q \land |\{k \in \Pi \mid \langle q.round, k \rangle
        \in q.received\}| \geq n-B\} = \{S \in \mathcal{P}(\Pi) \mid
        |S| \geq n-B\}$. We conclude by application of
        Theorem~\ref{cfreeHO}.
    \end{proof}

\section{Beyond Carefree Strategies: Reactionary Strategies}
\label{sec:general}

    Sometimes, carefree strategies are not enough to capture the
    subtleties of a Delivered predicate. Take the one corresponding to
    at most $F$ initial crashes for example: to make the most of this
    predicate, a strategy should remember from which processes it received
    a message, since it knows this process did not crash.
    A class of strategies which allows this is the class of
    reactionary strategies: they depend on messages from
    current and past rounds, as well as the round number.
    The only part of the local state these
    strategies cannot take into account is the set of messages received
    from "future" rounds, a possibility due to asynchrony.

    \subsection{Definition and Expressiveness Results}

    \begin{definition}[Reactionary Strategy]
        Let $f$ be a strategy, and $\forall q \in Q$, let $reac(q) \triangleq
        \langle q.round, \{ \langle r, k \rangle \in q.received \mid
        r \leq q.round \} \rangle$.
        $f$ is a \textbf{reactionary strategy} $\triangleq
        \forall q, q' \in Q: reac(q) = reac(q') \implies (q \in f \iff q' \in f)$.
        We write $\textit{Nexts}^R_f \triangleq \{reac(q) \mid q \in f\}$ for the set
        of reactionary states in $f$. This uniquely defines $f$.
    \end{definition}

    Even if reactionary strategies are more complex, we can still prove the same
    kind of results as for carefree ones, namely about validity
    and the existence of a reactionary strategy dominating all
    reactionary strategies.

    \begin{lemma}[Validity of Reactionary]
        \label{reacValid}
        Let \textit{PDel} be a Delivered predicate and
        $f$ a reactionary strategy.
        Then, $f$ is valid for $PDel \iff
        \forall \textit{CDel} \in \textit{PDel}, \forall r > 0, \forall j \in \Pi:
        \langle r, \{ \langle r',k \rangle \mid r' \leq r
        \land k \in \textit{CDel}(r',j)\} \rangle \in \textit{Nexts}^R_f$.
    \end{lemma}

    \begin{proof}
        $(\Rightarrow)$ Let $f$ be valid for \textit{PDel}.
        We show by contradiction that it satisfies the right-hand side of the
        above equivalence.
        \textbf{Assume} there is
        $\textit{CDel} \in \textit{PDel}, r > 0$ and $j \in \Pi$ such that
        $\langle r, \{ \langle r',k \rangle \mid r' \leq r
        \land k \in \textit{CDel}(r',j) \} \rangle \notin \textit{Nexts}^R_f$.
        Let $t$ be an earliest run of $f$ for
        \textit{CDel}.
        This is a run of $f$ by Lemma~\ref{trivialBis}.

        The sought contradiction is reached by proving that $t$ is invalid.
        To do so, we split according to two cases.
        \begin{itemize}
            \item During one of the first $r-1$ iterations of the
                $t$, there is some process which cannot change
                round. Let $r'$ be the smallest iteration where it happens,
                and $k$ be a process unable to change round at the $r'$-ith
                iteration.

                By minimality of $r'$, all processes arrive at round $r'$ in
                $t$; by definition of an earliest run, all messages for $k$
                from all rounds up to $r'$ are delivered before the \textit{next}
                for the iteration.
                Let $q$ the local state of $k$ at the
                first \textit{next} in the $r'$-ith iteration, and
                let $q'$ be any local state of $k$ afterward. 
                The above tells us that as long as $q'.round = q.round$,
                we have $reac(q) = reac( q')$ and thus $q' \notin f$.
                Therefore, $k$ can never change round while at round $r'$.

                We conclude that $t$ is invalid.
            \item For the first $r-1$ iterations, all processes change round.
                Thus, every one arrives at round $r$ in the $r-1$-ith iteration.
                By definition of
                an earliest run, all messages from rounds up to $r$ are delivered
                before any \textit{next} at round $r$. By hypothesis,
                $j$ cannot change round because its reactionary state --
                the combination of its rounds and the messages it received
                from past and current rounds -- is not
                in $\textit{Nexts}^R_f$.
                Let $q$ be the local state of $j$ at the
                first \textit{next} in the $r$-ith iteration, and
                let $q'$ be any local state of $j$ afterward. 
                The above tells us that as long as $q'.round = q.round$,
                we have $reac(q) = reac(q')$
                and thus $q' \notin f$. Therefore, $j$ can never
                change round while at round $r$.

                Here too, $t$ is invalid.
        \end{itemize}

        Either way, we reach a \textbf{contradiction} with the validity of $f$.

        $(\Leftarrow)$ Let \textit{PDel} and $f$ such that
        $\forall \textit{CDel} \in \textit{PDel},
        \langle r, \{ \langle r',k \rangle \mid r' \leq r
        \land k \in \textit{CDel}(r',j) \} \rangle \in \textit{Nexts}^R_f$.
        We show by contradiction that $f$ is valid.

        \textbf{Assume} the contrary: there is some
        $t \in \textit{runs}_f(\textit{PDel})$ which is invalid. Thus, there
        are some process blocked at a round forever in $t$. Let $r$ be the
        smallest such round, and $j$ be a process blocked at round $r$ in $t$.
        By minimality of $r$, all processes arrive at round $r$.
        By definition of a run of \textit{PDel}, there is a
        $\textit{CDel} \in \textit{PDel}$ such that $t$ is a run of \textit{CDel}.
        Thus, eventually all messages from all Delivered sets of $j$ up to round
        $r$ are delivered.

        From this point on, $f$ allows $j$ to change round by definition, and
        the fairness of \textit{next} imposes that $j$ does at some point.
        We conclude that $j$ is not blocked at round $r$ in $t$,
        which \textbf{contradicts} the hypothesis.
    \end{proof}

    There is also always a reactionary strategy dominating all reactionary
    strategies for a given Delivered predicate.
        Analogously to the
        carefree case, this stems from the fact that
        the strategy which only waits for the prefixes
        of Delivered collections in the predicate waits for more messages
        than any other valid reactionary strategy, by Lemma~\ref{reacValid}.

    \begin{theorem}[Always a Dominating Reactionary Strategy]
        \label{reacDom}
        Let \textit{PDel} a Delivered predicate and let $f_{rcDom}$ be the
        reactionary strategy defined by\\
        $\textit{Nexts}^R_{rcDom} =
        \left\{ \langle r, \{ \langle r',k \rangle \mid r' \leq r
        \land k \in \textit{CDel}(r',j)\} \rangle \mid
        r > 0 \land \textit{CDel} \in \textit{PDel}
        \right\}$.

        $f_{rcDom}$ dominates all reactionary strategies for \textit{PDel}.
    \end{theorem}

    \begin{proof}
        First, $f_{rcDom}$ is valid for \textit{PDel} by application of
        Lemma~\ref{reacValid}.

        As for domination, we also deduce from Lemma~\ref{reacValid} that
        $\forall f$ a valid reactionary strategy for \textit{PDel},
        $\textit{Nexts}^R_{f_{rcDom}} \subseteq \textit{Nexts}^R_f$ and thus
        $f_{rcDom} \subseteq f$. Therefore,
        $\forall q \in Q:
        q \in f_{rcDom} \implies q \in f$. This gives us
        $runs_{f_{rcDom}}(\textit{PDel}) \subseteq runs_f(\textit{PDel})$,
        and we conclude $PHO_{f_{rcDom}}(\textit{PDel})
        \subseteq PHO_f(\textit{PDel})$.
        We conclude that $f_{rcDom}$ dominates all valid carefree strategies
        for \textit{PDel}.
    \end{proof}

    \subsection{Example Dominated by Reactionary Strategy}

    To show the usefulness of reactionary strategies, we study the Delivered
    predicate corresponding to reliable communication and at most $F$
    initial crashes: either processes crash initially and no message
    of theirs is ever delivered, or they do not and all their messages will
    be delivered eventually.

    \begin{definition}[\textit{PDel}$_{ini}^F$]
        The Delivered predicate \textit{PDel}$_{ini}^F$ for
        at most $F$ initial crashes is:
        $\{ \textit{CDel} \in (\mathbb{N}^* \times \Pi) \mapsto \mathcal{P}(\Pi)
        \mid \exists \Sigma \subseteq \Pi: |\Sigma| \geq n-F \land
        \forall r > 0, \forall j \in \Pi: \textit{CDel}(r,j) = \Sigma \}$.
    \end{definition}

    As we mentioned above, it is possible to take advantage of the past
    by waiting in the current round for messages from processes which
    sent a message in a past round.

    \begin{definition}[Past complete strategy]
        The \textbf{past-complete strategy} $f_{pc}$ is defined by
        $\textit{Nexts}^R_{f_{pc}} =
        \{ \langle r, [1,r] \times \Sigma \rangle \mid r > 0 \land
        \Sigma \subseteq \Pi \land |\Sigma| \geq n-F \}$.
    \end{definition}

    \begin{lemma}[$f_{pc}$ Reactionary Dominates \textit{PDel}$_{ini}^F$]
        \label{fpcReacDom}
        $f_{pc}$ dominates all reactionary strategies for \textit{PDel}$_{ini}^F$.
    \end{lemma}

    \begin{proof}
        It follows from Theorem~\ref{reacDom}, because $Nexts^R_{f_{pc}}
        = \{ \langle r, \{ \langle r', k \rangle \mid r' \leq r
        \land k \in \Sigma\} \rangle \mid r > 0 \land
        \Sigma \subseteq \Pi \land |\Sigma| \geq n-F \}
        = \{\langle r, \{ \langle r', k \rangle \mid
        r' \leq r \land k \in \textit{CDel}(r',j)\} \rangle \mid
        r > 0 \land \textit{CDel} \in \textit{PDel}_{ini}^F \}$. The last
        equality follows from the fact that Delivered sets are always the
        same for all Delivered collection in \textit{PDel}$_{ini}^F$.
    \end{proof}

    Reactionary strategies can generate more complex and involved Heard-Of
    predicates than carefree ones. The one generated by $f_{pc}$ for
    \textit{PDel}$_{ini}^F$ is a good example: it ensures that all Heard-Of
    sets contains at least $n-F$ processes, and it also forces Heard-Of sets
    for a process to be non-decreasing, and for all rounds to eventually
    converge to the same Heard-Of set. This follows from the fact that a
    process can detect an absence of crash: if one message is received from
    a process, it will always be safe to wait for messages from this process
    as it did not crash and never will.
    Since there is no loss of message
    for non crashed processes, every one eventually receives a message from
    every process sending messages, and thus the Heard-Of sets converge.

    \begin{theorem}[Heard-Of Characterization of \textit{PDel}$_{ini}^F$]
        \label{pcCharac}
        $PHO_{f_{pc}}(\textit{PDel}_{ini}^F) =\\
        \left\{ \textit{cho} \in (\mathbb{N}^* \times \Pi) \mapsto \mathcal{P}(\Pi)
        ~\middle|~
        \left(
        \begin{array}{ll}
            & \forall r > 0, \forall j \in \Pi:
                \left(
                \begin{array}{ll}
                    & |\textit{cho}(r,j)| \geq n-F\\
                    \land & \textit{cho}(r,j) \subseteq \textit{cho}(r+1,j)\\
                \end{array}
                \right)\\
            \land & \exists \Sigma_0 \subseteq \Pi, \exists r_0 > 0,
            \forall r \geq r_0, \forall j \in \Pi: \textit{cho}(r,j) = \Sigma_0\\
        \end{array}
        \right)
        \right\}$
    \end{theorem}

    \begin{proof}
        First, we show $\subseteq$.
        Let $\textit{cho} \in PHO_{f_{pc}}(\textit{PDel}_{ini}^F)$
        and $t \in runs_{f_{pc}}(\textit{PDel}_{ini}^F)$ a run of $f_{pc}$
        implementing \textit{cho}.
        Then, by definition of $f_{pc}$, processes change
        round only when they have received at least $n-F$ messages from the current
        round and they have completed their past, which implies that
        $\forall r \in \mathbb{N}^*, \forall j \in \Pi :
        \left(
        \begin{array}{ll}
            & |\textit{cho}(r,j)| \geq n-F\\
            \land & \textit{cho}(r,j) \subseteq \textit{cho}(r+1,j)\\
        \end{array}
        \right)$.

        As for the last part of the Heard-Of predicate, notice that
        by definition of runs, there is a point in $t$ where all messages
        from round 1 have been delivered. From this point on, all processes 
        wait for all messages of processes which send them. Thus, there is a
        round -- the maximal round of a process when the last message from
        round 1 is delivered in $t$ -- from which the rounds are space-time
        uniform.

        Then, we show $\supseteq$.
        Let \textit{cho} a Heard-Of collection over $\Pi$ in the set on the right.
        We now build $t$, a special run of \textit{cho} in iterations:
        at each iteration, we deliver messages from Heard-Of sets and messages
        to complete the past of these sets then change round for all processes.
        That is, we deliver all messages
        from past rounds sent by the processes in the Heard-Of sets which had
        not yet been delivered. The initial state is also
        the usual one.

        It is a run because it satisfies all four constraints:
        \begin{itemize}
            \item $t$ has the right initial state by definition.
            \item We only gave the transitions and initial state, so
                the states satisfy the transition function.
            \item Every message is delivered in a round greater or equal to
                the one it was sent in.
            \item Messages are delivered only once.
        \end{itemize}
        It is also a run of $f_{pc}$ because each Heard-Of set contains
        at least $n-F$ processes by hypothesis and we complete the past of each
        process before playing any \textit{next}. We still have to show
        that it is a run of a Delivered collection in \textit{PDel}$_{ini}^F$.

        By hypothesis, there is a round $r_0$ from which the Heard-Of collection
        is space-time uniform -- all processes share the same Heard-Of set forever.
        Let $\Sigma_0$ this Heard-Of set. We then build \textit{CDel}
        such that $\forall r > 0, \forall j \in \Pi:
        \textit{CDel}(r,j) = \Sigma_0$. Because $\forall r > 0, \forall j \in \Pi:
        \textit{cho}(r+1,j) \supseteq \textit{cho}(r,j)$, and because we only
        ever deliver messages from processes in some Heard-Of sets, we know
        every message delivered in $t$ is in \textit{CDel}. On the other
        hand, all messages from \textit{CDel} are delivered eventually:
        \begin{itemize}
            \item If a message is in a Heard-Of set, it is delivered in the
                corresponding round.
            \item If it is not in a Heard-Of set, it is delivered at most
                in the first round after its sending round where the Heard-Of
                set contains this process. Such a round must exist by hypothesis,
                because only processes from $\Sigma_0$ get their messages delivered,
                and eventually all processes have $\Sigma_0$ as their Heard-Of set.
        \end{itemize}
        We conclude that $t \in runs_{f_{pc}}(\textit{PDel}_{ini}^F)$, and thus
        that $\textit{cho} \in PHO_{f_{pc}}(\textit{PDel}_{ini}^F)$.
    \end{proof}

    As for $f_{pc}$ dominating \textit{PDel}$_{ini}^F$, the argument is
    quite similar to the one for $f_{n-F}$ dominating \textit{PDel}$^F$,
    with a more subtle manipulation of Delivered collection because
    we need to take the past into account.

    \begin{theorem}[$f_{pc}$ Dominates \textit{PDel}$_{ini}^F$]
        $f_{pc}$ dominates \textit{PDel}$_{ini}^F$.
    \end{theorem}

    \begin{proof}
        Let $f$ be a valid strategy for \textit{PDel}$_{ini}^F$; we now prove
        that $f \prec_{\textit{PDel}_{ini}^F} f_{pc}$, that is
        $PHO_{f_{pc}}(\textit{PDel}_{ini}^F) \subseteq PHO_f(\textit{PDel}_{ini}^F)$.
        Thus, let $\textit{cho} \in PHO_{f_{pc}}(\textit{PDel}_{ini}^F)$.
        We now build $t$, a special run of \textit{cho}:
        at each round, we deliver messages from Heard-Of sets and messages
        to complete the past of these sets, then all processes change round.
        That is, we deliver all messages
        from past rounds sent by processes in the current Heard-Of sets which had
        not yet been delivered. The initial state is also
        the usual one.

        It is a run because it satisfies all four constraints:
        \begin{itemize}
            \item $t$ has the right initial state by definition.
            \item We only gave the transitions and initial state, so
                the states satisfy the transition function.
            \item Every message is delivered in a round greater or equal to
                the one it was sent in.
            \item Messages are delivered only once.
        \end{itemize}

        We want to show that $t$ is a run of $f$; we do so
        by contradiction.
        \textbf{Assume} it is not:
        because $t$ is infinite, the problem
        is either the infinite fairness of \textit{next} or a \textit{next} done
        when $f$ does not allow it. Each  \textit{next} transition being played
        an infinite number of times, the only possibility left is
        the second one: some \textit{next} transition in $t$ is done while
        $f$ does not allow the corresponding process to change round.
        There thus exists a smallest $r$ such that some process $j$ is not allowed
        by $f$ to change round when \textit{next}$_j$ is played at round $r$
        in $t$. At this point, we know that the past of $j$ is complete
        and that it received all messages from \textit{cho}$(r,j)$. By
        Theorem~\ref{pcCharac}, we know that $\forall r' < r: \textit{cho}(r',j)
        \subseteq \textit{cho}(r,j)$. Thus, when we complete the past
        of $j$, we deliver all messages at each round $\leq r$ from the processes
        in $\textit{cho}(r,j)$.

        We now build a Delivered collection $\textit{CDel}_{block}$ such that
        $\forall r' > 0, \forall k \in \Pi:
        \textit{CDel}_{block}(r',k) = \textit{cho}(r,j)$.
        This is in $\textit{PDel}_{ini}^F$ because $\textit{cho}(r,j)$ contains
        at least $n-F$ processes by Theorem~\ref{pcCharac}.
        Finally, we build $t_{block}$, an earliest run of $f$
        for \textit{CDel}$_{block}$.
        By Lemma~\ref{trivialBis}, we know $t_{block}$ is a run of $f$.
        We then have two possibilities.
        \begin{itemize}
            \item During one of the first $r-1$ iterations of
                $t_{block}$, there is some process which cannot change
                round. Let $r'$ be the smallest iteration where it happens,
                and $k$ be a process unable to change round at the $r'$-ith
                iteration.

                By minimality of $r'$, all processes arrive at round $r'$,
                and by symmetry of $\textit{CDel}_{\textit{block}}$ they all
                receive the same messages as $k$. Thus, all processes
                are blocked at round $r'$, there are no more next or deliveries,
                and $t_{block}$ is therefore invalid.
            \item For the first $r-1$ iterations, all processes change round.
                Thus, every one arrives at round $r$ at the $r$-ith iteration.
                By definition of
                an earliest run, all messages from the round are delivered
                before any \textit{next}. The symmetry of
                $\textit{CDel}_{\textit{block}}$ also ensures that every process
                received the same messages, that is all messages from
                round $< r$ and all messages from \textit{cho}$(r,j)$.
                These are exactly the messages received by $j$ in $t$ at
                round $r$.
                But by hypothesis, $j$ is blocked in this state in $t$.
                We thus deduce that all processes are blocked at round $r$
                in $t_{block}$, and thus that it is an invalid run.
        \end{itemize}

        Either way, we deduce that $f$ is invalid, which is a
        \textbf{contradiction}.
    \end{proof}

    \subsection{When The Future Serves}

    In the above, we considered cases where the dominating strategy is at
    most reactionary: only the past and present rounds are useful
    for generating Heard-Of collections. But messages
    from future rounds serve in some cases.
    We give an example, presenting only the intuition.

    \begin{definition}[\textit{PDel}$_{lost}^1$]
        The Delivered predicate \textit{PDel}$_{lost}^1$ corresponding to
        at most 1 message lost is:
        $\{ \textit{CDel} \in (\mathbb{N}^* \times \Pi) \mapsto \mathcal{P}(\Pi)
        \mid \sum\limits_{r>0, j \in \Pi} |\Pi \setminus \textit{CDel}(r,j)|
        \leq 1\}$.
    \end{definition}

    Application of our results on carefree strategy shows that the
    carefree strategies dominating all carefrees for this predicate is
    $f_{n-1}$. Similarly, looking at the past only allows processes to
    wait for $n-1$ messages because one can always deliver all messages from
    the past, and then the loss might be a message from the current round.
    If we look at the messages from the next round, on the other hand, we
    can ensure that at each round, at most one message among all processes
    is not delivered on time.

    \begin{definition}[Asymmetric Strategy]
        Let $\textit{after}: Q \mapsto \mathcal{P}(\Pi)$ such that $\forall q \in Q:
        \textit{after}(q) = \{ k \in \Pi \mid
        \langle q.round + 1, k \rangle \in q.received\}$,
        and \textit{cfree} as in Definition~\ref{carefree}.\\
        Define $f_{asym} \triangleq
        \left\{
        q \in Q ~\middle|~
        \begin{array}{ll}
            & \textit{cfree}(q) = \Pi\\
            \lor & (|\textit{after}(q)| = n-1 \land |\textit{cfree}(q)| = n-1\}
        \end{array}
        \right\}$.
    \end{definition}

    Intuitively, this strategy is valid because at each round and for each
    process, only two cases exist:
 either no message for this process at this round is lost, and
            it receives a message from each process;
         or one message for this process is lost at this round, and
            it only receives $n-1$ messages. But all other processes
            receive $n$ messages, thus change round and send
            their message from the next round. Since the one loss already
            happened, all these messages are delivered, and the original
            process eventually receives $n-1$ messages from the next round.

            This strategy also ensures that at most one process per round receives
    only $n-1$ messages on time -- the others must receive all messages. This
    vindicates the value of messages from future rounds for some Delivered predicates,
    such as the ones with asymmetry in them.

\section{Related Works}
\label{sec:relatedworks}

    \subparagraph{Rounds Everywhere}

    Rounds in message-passing algorithms date at least back
    to their use by Arjomandi et al.~\cite{ArjomandiFirstRound} as
    a synchronous abstraction of time complexity. Since then, they are
    omnipresent in the literature. First, the number of rounds taken by a
    distributed computation is a measure of its complexity.
    Such round complexity was even developed into a full-fledged
    analogous of classical complexity theory by Fraigniaud et
    al.~\cite{FraigniaudComplexity}. Rounds also serve as stable
    intervals in the dynamic network model championed by Kuhn and
    Osham~\cite{KuhnDynamic}: each round corresponds
    to a fixed communication graph, the dynamicity following
    from possible changes in the graph from round to round. Finally,
    many fault-tolerant algorithms are structured in rounds,
    both synchronous~\cite{FisherAlgo} and asynchronous ones~\cite{ChandraCHT}.

    Although we only study message-passing models in this article, one
    cannot make justice to the place of rounds in distributed computing without
    mentioning its even more domineering place in shared-memory models.
    A classic example is the structure of executions underlying the
    algebraic topology approach pioneered by Herlihy and
    Shavit~\cite{HerlihyTopology}, Saks and Zaharoglou~\cite{SaksTopology},
    and Borowsky and Gafni~\cite{BorowskyTopology}.

    \subparagraph{Abstracting the Round}

    Gafni~\cite{GafniRRFD} was the first to attempt the unification
    of all versions of rounds. He introduced
    the Round-by-Round Fault Detector abstraction,
    a distributed module analogous to a failure detector which outputs
    a set of suspected processes. In a system using RRFD, the end condition
    of rounds is the reception of a message from every process not suspected
    by the local RRFD module;
    communication properties are then defined as predicates on the
    output of RRFDs. Unfortunately, this approach does not suit our
    needs: RRFDs do not ensure termination of rounds, while we require it.

    Next, Charron-Bost and Schiper~\cite{CharronBostHO} took a dual
    approach to Gafni's with the Heard-Of Model. Instead of specifying
    communication by predicates on a set of suspected processes, they
    used Heard-Of predicates: predicates on a collection of
    Heard-Of sets, one for each round $r$ and each process $j$, containing
    every process from which $j$ received the message sent in
    round $r$ before the end of this same round.
    This conceptual shift brings two advantages: a purely
    abstract characterization of message-passing models
    and the assumption of infinitely many rounds, thus of round termination.

    But determining which model implements a given Heard-Of
    predicate is an open question. As mentioned in
    Mari{\'{c}}~\cite{MaricCutoff},
    the only known works addressing it, one by Hutle and
    Schiper~\cite{HutleComPredicates} and the other by Dr\u{a}goi et
    al.~\cite{DragoiPSYNC}, both limit themselves
    to very specific predicates and partially synchronous system models.

\section{Conclusion and Perspectives}
\label{sec:conclusion}

    We propose a formalization for characterizing Heard-Of predicate of
    an asynchronous message-passing model through Delivered
    predicates and strategies. We also show its relevance,
    expressivity and power: it
    allows us to prove the characterizations from Charron-Bost and
    Schiper~\cite{CharronBostHO}, to show the existence of characterizing
    predicates for large classes of strategies as well as the form
    of these predicates. Yet there are two aspects of this research
    left to discuss: applications and perspectives.

    First, what can we do with this characterizing Heard-Of predicate?
    As mentioned above, it gives the algorithm designer a concise
    logical formulation of the properties on rounds a given model
    can generate. Therefore, it allows the design of algorithms at a higher level
    of abstraction, implementable on any model which can generate the
    corresponding Heard-Of predicate.
    The characterizing predicate is also crucial to verification: it bridges
    the gap between the intuitive operational model and its formal
    counterpart. To verify a round-based algorithm for a given message-passing
    model, one needs only to check if its correctness for the characterizing
    predicate. If it is the case, we have a correct implementation by combining
    the algorithm with a dominating strategy; if it is not, then the algorithm
    will be incorrect for all predicates generated by the model.

    Finally, it would be beneficial to
    prove more results about the existence of a dominating strategy,
    as well as more conditions for the dominating strategy to be
    in a given class. There is also space for exploring
    different Delivered predicates. For example, some of our results suppose
    that the predicate contains the total Delivered collection; what can we
    do without this assumption? Removing it means that faults are certain
    to occur, which is rarely assumed. Nonetheless,
    it might be interesting to study this case, both as a way to strengthen
    our results, and because forcing failures might be relevant when modelling
    highly unreliable environments such as the cloud or natural settings.
    Another viable direction would be to add oracles to the processes,
    giving them additional information about the Delivered collection, and
    see which Heard-Of predicates can be generated. These oracles might for
    example capture the intuition behind failure detectors.

%%
%% Bibliography
%%

%% Either use bibtex (recommended),

\bibliography{references}

%% .. or use the thebibliography environment explicitly

\end{document}